\documentclass{IEEEtran}
\usepackage[latin1]{inputenc}
\usepackage{amsmath, amsfonts, amssymb, amsthm, graphicx, color}
\usepackage{epsfig}
\usepackage{fancyhdr}
\usepackage{framed}
\usepackage{floatrow}
\usepackage{mdframed}
\usepackage{url}
\usepackage{mathabx}
\usepackage{cfr-lm}
\usepackage{cite}
\usepackage{multicol}
\usepackage{booktabs}
\usepackage{relsize}
\usepackage{tabularx}
\usepackage{tabulary}
\usepackage{hyperref}

\usepackage[bottom]{footmisc}

\usepackage[]{algorithm2e}
\usepackage{algorithmic}
\floatname{algorithm}{Procedure}

\newcounter{abci}\renewcommand{\theabci}{\alph{abci}}

\theoremstyle{plain}
\newtheorem{lem}{Lemma}
\newtheorem{thm}{Theorem}
\newtheorem{prop}{Proposition}

\theoremstyle{definition}

\newcommand{\Tr}{{\rm Tr}}

\author{
    \IEEEauthorblockN{Evgenia Chunikhina\IEEEauthorrefmark{2}, Paul Logan\IEEEauthorrefmark{3}, Yevgeniy Kovchegov\IEEEauthorrefmark{4}, Anatoly Yambartsev\IEEEauthorrefmark{1}}, Debashis Mondal\IEEEauthorrefmark{3}\IEEEauthorrefmark{5}, Andrey Morgun\IEEEauthorrefmark{6}\\ 
    \IEEEauthorblockA{\IEEEauthorrefmark{2}Data Science, Pacific University, Forest Grove, OR, USA  \quad Email: chunikhe.vazhno@gmail.com}\\
    \IEEEauthorblockA{\IEEEauthorrefmark{3}Statistics Department, Oregon State University, Corvallis, OR, USA}\\
    \IEEEauthorblockA{\IEEEauthorrefmark{4}Department of Mathematics, Oregon State University, Corvallis, OR, USA}\\
    \IEEEauthorblockA{\IEEEauthorrefmark{1}IME, University of S\~{a}o Paulo, S\~{a}o Paulo, Brazil}\\
    \IEEEauthorblockA{\IEEEauthorrefmark{5}Dept. of Mathematics and Statistics, Washington University in St. Louis, St. Louis, MO, USA}\\
    \IEEEauthorblockA{\IEEEauthorrefmark{6}College of Pharmacy, Oregon State University, Corvallis, OR, USA}\\
}

\title{The C-SHIFT algorithm for normalizing covariances}

\date{}

\begin{document}
\maketitle



\begin{abstract}
Omics technologies are powerful tools for analyzing patterns in gene expression data for thousands of genes. Due to a number of systematic variations in experiments, 
the raw gene expression data is often obfuscated by undesirable technical noises. Various normalization techniques were designed in an attempt to remove these 
non-biological errors prior to any statistical analysis. 
One of the reasons for normalizing data is the need for recovering the covariance matrix used in gene network analysis.
In this paper, we introduce a novel normalization technique, called the covariance shift (C-SHIFT) method.
This normalization algorithm uses optimization techniques together with the blessing of dimensionality philosophy and energy minimization hypothesis
for covariance matrix recovery under additive noise (in biology, known as the bias). Thus, it is perfectly suited for the analysis of logarithmic gene expression data.
Numerical experiments on synthetic data demonstrate the method's advantage over the classical normalization techniques.
Namely, the comparison is made with Rank, Quantile, cyclic  LOESS (locally estimated scatterplot smoothing), and MAD (median absolute deviation) normalization methods.
We also evaluate the performance of C-SHIFT algorithm on real biological data. 
\end{abstract}

Gene expression analysis plays an important role in genomic research. Several omics technologies such as RNAseq and microarrays allow for the collection of massive amounts of simultaneous measurements of gene expression levels of thousands to tens of thousands of genes. Analyzing different patterns of gene expressions helps to gain insight into complex biological phenomena such as development, aging, onset and progression of diseases, and cellular response/reaction to drugs/treatments. Although new technologies are constantly developing, it is well known that all of them generate some technical noise which affects the measured gene expression levels \cite{hartemink2001maximum,scherer2009batch}. To extract accurate biological information it becomes necessary to normalize the data to filter out/compensate for these non-biological noises/errors. Normalization is a crucial pre-processing step in the gene expression data analysis. The gene expression data will vary significantly after different normalization methods. Thus, the results of further data analysis (e.g. gene expression network) will be critically dependent on a choice of a normalization technique.
A variety of normalization procedures have been used on gene expression data sets. See \cite{pradervand2009impact,rao2008comparison,quackenbush2002microarray,bilban2002normalizing,bolstad2003comparison,park2003evaluation,saccenti2017correlation,liu2019normalization} and reference therein for a review and comparison of current normalization strategies.
In this paper we develop a novel normalization technique, called the covariance shift (C-SHIFT) method, and compare it to the following well known normalization methods used in large scale data analysis: Rank, Quantile, cyclic  LOESS (locally estimated scatterplot smoothing), and MAD (median absolute deviation). See \cite{qiu2013impact, bolstad2003comparison, amaratunga2001analysis, quackenbush2002microarray} and references therein for more details on the above listed normalization methods.
There is an important distinction: while Rank, Quantile, LOESS and other normalizations normalize the data, C-SHIFT algorithm normalizes the covariances.
The need to normalize the covariances is caused by the presence of bias.


\subsection{Bias.} Consider a situation where the gene expression data is subjected to multiplicative noise (aka bias).
Let $M$ be the number of genes and $N$ be the number of measurements.
Next, we let $X_n^{(i)}$ denote the true gene expression, where subscript index $n$ stands for the $n$-th gene in the network 
and the superscript index $i$ stands for the $i$-th measurement.  The observed gene expression, denoted by $\widetilde{X}_n^{(i)}$, is different from $X_n^{(i)}$
due to all gene expressions in the $i$-th measurement being distorted by a multiplicative noise $W^{(i)}$, i.e.,
\begin{equation}\label{eqn:multNoise}
\widetilde{X}_n^{(i)}=W^{(i)} X_n^{(i)},
\end{equation}
where random variables $X_n^{(i)}$ are independent of the variable $W^{(i)}$. 
Additionally random variables $W^{(i)}$ ($i=1,\hdots,N$) are assumed to be independent and identically distributed (i.i.d.).
Here, both the observed and the true gene expressions are positive, i.e., $X_n^{(i)}>0$ and $W^{(i)}>0$. 

\medskip
\noindent
In biology, the multiplicative noise $W^{(i)}$ is referred to as the bias. 
The bias is prompted by random events causing an error in the measurement of the total amount of RNA.
Such random events are often related to different levels of tissue preservation in different samples that leads to variability of RNA degradation.
Consequently, this leads to an RNA detection problem.
Additionally, there are other technical reasons for an error in the measurement of the total amount of RNA in a given sample that may lead to a bias in \eqref{eqn:multNoise}.
All other noise (e.g. misreading parts of RNA) goes into the variable $X_n^{(i)}$.

\medskip
\noindent
The multiplicative noise in \eqref{eqn:multNoise} implies the corresponding additive noise (bias)
in the logarithimic gene expression data:
\begin{equation}\label{eqn:additiveNoise}
\widetilde{Y}_n^{(i)}=Y_n^{(i)}+V^{(i)},
\end{equation}
where  we let
$\widetilde{Y}_n^{(i)}:=\log\widetilde{X}_n^{(i)}$, $~Y_n^{(i)}:=\log X_n^{(i)}$, and $~V^{(i)}:=\log W^{(i)}$.

\medskip
\noindent
\subsection{Impact of bias on covariances and correlations.}
While the bias may not appear critical, they are known to cause significant problems in 
the analyses of gene correlation structure. Specifically, this phenomenon is known to cause the disappearance of the large magnitude negative correlations in the observed biological data, $\widetilde{X}_n$ and $\widetilde{Y}_n$, which hampers the ability to perform certain types of statistical data analysis, such as the false discovery rate (FDR) method.

The bias, whether multiplicative as in \eqref{eqn:multNoise} or additive as in \eqref{eqn:additiveNoise}, causes the correlations to be shifted away from $-1$.
In partiocular, the independent additive noise in \eqref{eqn:additiveNoise} implies an increase of theoretical covariance as
\begin{equation}\label{eqn:covLog}
Cov(\widetilde{Y}_n, \widetilde{Y}_m)=Cov(Y_n, Y_m)+\omega,
\end{equation}
where $\omega=Var(V)>0$. Consequently, the correlations in the logarithmic data are equal to
\begin{equation}\label{eqn:corrLog}
corr(\widetilde{Y}_n, \widetilde{Y}_m) =\frac{Cov(Y_n,Y_m)+\omega}{\sqrt{\Big(Var(Y_n)+\omega \Big)\Big(Var(Y_m)+\omega \Big)}}.
\end{equation}
If $Cov(Y_n,Y_m)$ is negative, by adding $\omega>0$ in the numerator and the denominator, we obtain
$$corr(\widetilde{Y}_n, \widetilde{Y}_m) > corr(Y_n, Y_m).$$
Hence, the disappearance of large magnitude negative correlations.

\medskip
\noindent
The purpose of the covariance shift (C-SHIFT) algorithm developed in this current manuscript is to normalize covariances in the logarithmic data and restore the correlations, 
thus offsetting the impact of the additive bias in \eqref{eqn:additiveNoise}.
Consequently, the comparison of C-SHIFT covariance normalization algorithm with methods of normalizing data such as Rank, Quantile, or LOESS can only be done
in terms of the effectiveness of recovering true empirical correlations. This comparison will be implemented on synthetic data in Section \ref{sec:num}
and on real biological data in Section \ref{sec:real}.

\medskip
\noindent
The problem of improving the existing and developing new normalization methods is very important for scientists working with biological data.  
The fact that normalization alters the data-correlation structure was stated in Saccenti \cite{saccenti2017correlation}. 
Besides \cite{saccenti2017correlation} gives a comprehensive overview of normalization methods.
In Bolstad {\it et al.} \cite{bolstad2003comparison} the authors compare three complete data normalization methods (cyclic LOESS,  contrast based method, and quantile), that make use of data from all arrays in an experiment, with two methods that make use of a baseline array. The comparison was done on two publicly available datasets with the results favoring the complete data methods.
For more on the normalization methods, see \cite{amaratunga2001analysis,fan2005semilinear,hartemink2001maximum,hu2007enhanced,quackenbush2002microarray,reilly2003method,smyth2003normalization,cheng2016icn,cheng2016crossnorm}.



\subsection{Paper structure and workflow diagram.}
The paper is organized as follows. In Section \ref{sec:theor}, we formulate C-SHIFT method from the underlying theoretical considerations. 
Section \ref{sec:theor} also contains Lemma \ref{lem:alpha}, Lemma \ref{lem:Hessian}, and Theorem \ref{thm:Convex}, necessary for the optimization approach to work. 
The pseudocode for the C-SHIFT algorithm is given in Section \ref{sec:alg}.
Section \ref{sec:num} contains numerical experiments on two synthetic datasets, 
one generated using random covariance method (RCM) and another generated with a cascade method. 
The effectiveness of recovering the true empirical correlation matrices is evaluated for C-SHIFT, Rank, Quantile, LOESS, and MAD.
Section \ref{sec:real} evaluates the outcomes of correlation recovery using six real biological datasets from GEO depository.
A comparison is made of C-SHIFT with Rank, Quantile, and LOESS.
The results and future directions are discussed in Section \ref{sec:discuss}.
Finally, Section \ref{sec:proofs} contains the proofs 
of Proposition \ref{prop:1C1}, Lemma \ref{lem:alpha}, Lemma \ref{lem:Hessian}, and Theorem \ref{thm:Convex}.

\medskip
\noindent
In Table \ref{tablenotation} we clarify the notation and introduce a few important quantities that will be used throughout the paper.
Workflow diagram can be found in Fig.~\ref{fig:FlowChart}.

\begin{table*}[htbp]\caption{Table of Notations }
\centering 
\begin{tabular}{r c p{10cm} }
\toprule
$N$ & $\triangleq$ & number of measurements\\
$M$ & $\triangleq$ & number of genes\\
$X_n^{(i)}$ & $\triangleq$ & $i$-th measurement of true gene expression of the $n$-th gene, $X_n^{(i)}>0$ \\
$W^{(i)}$ & $\triangleq$ & i.i.d. multiplicative noise variable a.k.a bias, $W^{(i)}>0$, independent of $X_n^{(i)}$ \\
$\tilde X_n^{(i)}$ & $\triangleq$ & $i$-th measurement of observed gene expression of the $n$-th gene, $\tilde X_n^{(i)} = W^{(i)}X_n^{(i)}$ \\
$Y_n^{(i)}$ & $\triangleq$ & $i$-th measurement of true logarithmic gene expression of the $n$-th gene, $Y_n^{(i)} = \log{X_n^{(i)}}$ \\ 
$V^{(i)}$ & $\triangleq$ & additive noise variable (bias), $V^{(i)} = \log{W^{(i)}}$ \\
$\tilde Y_n^{(i)}$ & $\triangleq$ & $i$-th measurement of the observed logarithmic gene expression of the $n$-th gene, $\tilde Y_n^{(i)} = \log{\tilde X_n^{(i)}} = Y_n^{(i)}+V^{(i)}$ \\
$\tilde C$ & $\triangleq$ & empirical covariance matrix of observed logarithmic gene expression data $\tilde Y_n^{(i)}$, $\tilde C = \left(\widehat{Cov}(\tilde Y_n, \tilde Y_m)\right)_{n,m} \in \mathbb{R}^{M \times M}$\\
$ C$ & $\triangleq$ & empirical covariance matrix of true logarithmic gene expression data $Y_n^{(i)}$, $C = \left(\widehat{Cov}(Y_n,Y_m)\right)_{n,m} \in \mathbb{R}^{M \times M}$\\
$\hat a_n $ & $\triangleq$ & empirical covariance between true logarithmic gene expression data $Y_n$ and additive bias $V$, $\hat a_n = - \widehat{Cov}(Y_n,V)$, $n = \overline{1,M}$ \\
$\hat \omega $ & $\triangleq$ & empirical variance of additive bias $V^{(i)}$, $\hat \omega  = \widehat{Var}(V) > 0$\\
\bottomrule
\end{tabular}
\label{tablenotation}
\end{table*}

\begin{figure}
\begin{center}
\includegraphics[width=\linewidth]{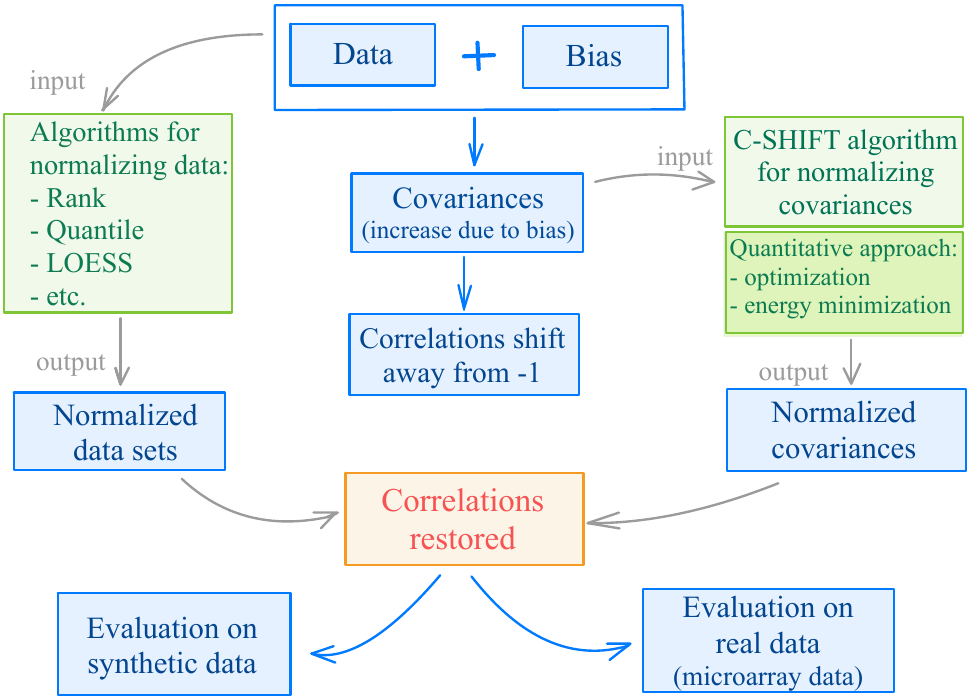}
\caption{Workflow diagram.}
\label{fig:FlowChart}
\end{center}
\end{figure}

\section{Theoretical derivations}\label{sec:theor}

Denote by $\widehat{Cov}$ the empirical covariances taken over $N$ samples for each of $\binom{M}{2}$ pairs of genes. 
Similarly, let $\widehat{Var}$ denote the empirical variance.
Then, equation \eqref{eqn:additiveNoise} yields the observed empirical covariance
\begin{equation}\label{eqn:covLogEmpirical}
\widehat{Cov}(\widetilde{Y}_n, \widetilde{Y}_m)=\widehat{Cov}(Y_n,Y_m)-\hat{a}_n-\hat{a}_m+\hat{\omega}
\end{equation}
for all pairs of gene indices $n$ and $m$, where $\hat{a}_n=-\widehat{Cov}(Y_n,V)$ for all $n=1,\hdots,M$, and $\hat{\omega}=\widehat{Var}(V)>0$.
As is often the case, $\hat{\omega}$ can be very large relative to the values of $\hat{a}_n$, causing the disappearance of the large magnitude negative correlations
in empirical data. 

The goal of the covariance shift (C-SHIFT) normalization method introduced here is the recovery of the true empirical covariances $\widehat{Cov}(Y_n,Y_m)$
and the respective true empirical correlations in the case of the logarithmic gene expression data or any other situations with additive noise as in \eqref{eqn:additiveNoise}.

\medskip
\noindent
Let 
$\widetilde{C}=\big(\widehat{Cov}(\widetilde{Y}_n, \widetilde{Y}_m)\big)_{n,m}$ 
be the empirical covariance matrix of the observed data $\widetilde{Y}_n^{(i)}$, and let 
$\,C=\Big(\widehat{Cov}(Y_n, Y_m)\Big)_{n,m}\,$ 
be the empirical covariance matrix of the cleaned data $Y_n^{(i)}$ (i.e., the true empirical covariance) that we desire to recover.
Formula \eqref{eqn:covLogEmpirical} rewritten in the matrix form states
\begin{equation}\label{eqn:C}
C =\widetilde{C} + \hat{a}{\bf 1}\!^T+ {\bf 1} \widehat{a}^T-\hat{\omega} \, {\bf 1\!1}\!^T,
\end{equation}
where $\hat{a}=\big(\hat{a}_1, \hdots,\hat{a}_M\big)^T$,
and ${\bf 1}$ denotes the column vector of $1$'s, hence ${\bf 1\!1}\!^T$ is a square matrix of $1$'s.

Our goal here is to estimate $\hat{a}$ and $\hat{\omega}$ in \eqref{eqn:C}, and thus recover the true empirical covariance matrix $C$.
We assume large dimension $M$. 
There will be two cases. 

{\it Case I:} If $\det(\widetilde{C})=0$ (e.g. $N < M$), we make a small perturbation of the diagonal entries of $\widetilde{C}$ (the variances) resulting in 
a new covariance matrix being positive definite whose smallest eigenvalue is still very close to zero. Next, we use energy minimization 
to estimate $\hat{a}_n$ and $\hat{\omega}$ in \eqref{eqn:C}. 

{\it Case II:} If $\widetilde{C}$ is positive definite (full rank), our approach exploits the phenomenon sometimes referred to as the 
{\it curse of dimensionality} \cite{bach2017breaking,radovanovic2010hubs} and sometimes as the {\it blessing of dimensionality}  \cite{gorban2018blessing,Donoho00high-dimensionaldata,kainen1997utilizing},
postulating that in higher dimensions almost all data points are located near extrema (i.e., in the outer shell)\footnote{In this paper we will refer to the phenomenon as the blessing of dimensionality rather than the curse of dimensionality.}. 
In other words, for large $M$, we anticipate the smallest eigenvalue of $C$ to be near zero.
As a rigorous bound, we observe that if some of the correlations $corr(Y_n, Y_m)$ are located in $[-1,\delta -1]$ interval, then 
the smallest eigenvalue of $C$ is located within $\big[0, \delta \min\limits_n  \widehat{Var}(Y_n)\big]$ interval.
Thus, as in Case I, under the blessing of dimensionality assumption, we again use energy minimization for estimating $\hat{a}_n$ and $\hat{\omega}$.

\medskip
\noindent
Next, we will need the following result.
\begin{prop}\label{prop:1C1}
Suppose $\mathcal{M}$ is a symmetric positive definite square matrix, and let
$$v^* := \max\big\{v:\, \mathcal{M}\!-\!v\, {\bf 1\!1}\!^T \text{  is positive semidefinite}\,\big\}.$$
Then,
$$v^* = \frac{1}{{\bf 1}\!^T \,\mathcal{M}^{-1} \, {\bf1}}.$$
\end{prop}
\noindent
The proof of Proposition \ref{prop:1C1} is given in Section \ref{sec:proofs}.

 \medskip
 \noindent
Suppose the empirical covariance matrix $\widetilde{C}$ is positive definite, i.e.,  $\widetilde{C}$ is of full rank.
Consider values of a column vector $\alpha=(\alpha_1, \hdots,\alpha_M)^T$ such that 
$$\widetilde{C}+ \alpha{\bf 1}\!^T+ {\bf 1}\!^T \alpha$$ is positive definite. 
If we let 
\begin{equation}\label{eqn:valpha}
v(\alpha):=\frac{1}{{\bf 1}\!^T \,\big(\widetilde{C}+ \alpha{\bf 1}\!^T+ {\bf 1} \alpha^T \big)^{-1} \, {\bf1}},
\end{equation}
then Prop.~\ref{prop:1C1} implies
\begin{equation}\label{eqn:Calpha}
C_\alpha:=\widetilde{C}+ \alpha{\bf 1}\!^T+ {\bf 1} \alpha^T-v(\alpha){\bf 1\!1}\!^T
\end{equation}
is positive semidefinite with $\det(C_\alpha)=0$.

\medskip
\noindent
Next, recall the quantities $\hat{a}$ and $\hat{w}$ in \eqref{eqn:C}.
If $\widetilde{C}$ is rank deficient,
we perturb its diagonal entries by adding small positive (random or deterministic) values, and
if $\widetilde{C}$ has full rank, we assume the blessing of dimensionality phenomenon holds.
Thus, in either case, we work under the assumption that $\widetilde{C}$ is positive definite with its smallest eigenvalue located near zero.
Then, Prop.~\ref{prop:1C1} implies $\hat{w} \approx v(\hat{a})$, where $v(\alpha)$ is as defined in \eqref{eqn:valpha}.
Therefore, letting $\alpha=\hat{a}$ in \eqref{eqn:Calpha}, we will have $C_{\hat{a}}$ approximating $C$ expressed as in \eqref{eqn:C}.

\medskip
\noindent
Now, for a matrix $X$, let $\|X\|_F$ denote the Frobenius norm of $X$ and let $\mathcal{E}(X)=\frac{1}{2} \|X\|_F^2$ be 
the energy function. 
Our next assumption states that $\hat{a}$ can be estimated by the minimizer $\alpha^*$ of the energy function $\mathcal{E}(C_\alpha)$, 
i.e., we estimate $\hat{a}$ with 
$$\alpha^*=\text{argmin}\|C_\alpha\|_F.$$
The assumption is additionally justified by the observation that a random adjustment of the covariance via an additive noise (bias) as in \eqref{eqn:covLogEmpirical}
will result in an energy increase, i.e., $\mathcal{E}(\widetilde{C})>\mathcal{E}(C)$. 

\medskip
\noindent
Matrix $C_{\alpha^*}$ will approximate $C_{\hat{a}}$, which, in turn, approximates the desired true empirical covariance matrix $C$.
The covariance shift (C-SHIFT) algorithm works as follows: it uses optimization algorithms to estimate $\alpha^*$ and outputs $C_{\alpha^*}$ as an estimate for $C$.
See Algorithm \ref{alg:C-SHIFT} in Section \ref{sec:alg}.

\medskip
\noindent
The following lemma yields a close form expression for the gradient $\nabla  \|C_\alpha\|_F^2$ that will be used to estimate $\alpha^*$ which minimizes $\|C_\alpha\|_F$. 
\begin{lem}\label{lem:alpha}
Suppose the empirical covariance matrix  $\widetilde{C}$ is of full rank, and the quantities $C_\alpha$ and  $v(\alpha)$ are as in \eqref{eqn:Calpha} and \eqref{eqn:valpha}. 
Then, the gradient of the Frobenius norm squared is given by
\begin{align}\label{eqn:gradient}
\frac{1}{4} \nabla  \|C_\alpha\|_F^2 =& \,M\alpha + \widetilde{C}{\bf 1}+[a-M\,v(\alpha)]{\bf 1} \\ &+ [M^2\,v^2(\alpha)-c\,v(\alpha)-2Ma\,v(\alpha)]A_\alpha^{-1} {\bf 1}, \nonumber 
\end{align}
where $\|\cdot\|_F$ denotes the Forbenius norm,  and we let 
\begin{equation}\label{eqn:Aalpha}
A_\alpha \nonumber :=\widetilde{C}+ \alpha{\bf 1}\!^T+ {\bf 1} \alpha^T, \quad 
c:={\bf 1}\!^T \!\widetilde{C} {\bf1}, \quad a:={\bf 1}\!^T \alpha.
\end{equation}  
\end{lem}
\noindent
The proof of Lemma \ref{lem:alpha} is given in Section \ref{sec:proofs}.

\medskip
\noindent
First, we observe that $C_\alpha$ is invariant under the addition of multiples of ${\bf 1}$. Thus, without loss of generality, we restrict the domain to a hyperplane $a={\rm Const}$.
Next, we notice that ${\bf 1}^T\nabla \|C_\alpha\|_F^2=0$ in \eqref{eqn:gradient}. Thus, in the gradient descent method, the value of $a$ remains constant, i.e., 
throughout the algorithm, vector $\alpha$ remains on the same hyperplane ${\bf 1}^T\alpha={\rm Const}$.

\medskip
\noindent
In our next lemma, we find the Hessian of $\|C_\alpha\|_F^2$. 
When minimizing $\|C_\alpha\|_F$ (equivalently, $\|C_\alpha\|_F^2$) both the gradient and the Hessian of $\|C_\alpha\|_F^2$ 
are inputted in the optimization algorithm such as trust-region or gradient descent.
\begin{lem}\label{lem:Hessian}
Suppose the empirical covariance matrix  $\widetilde{C}$ is of full rank, and the quantities $C_\alpha$,  $v(\alpha)$, and $A_\alpha$ are as in \eqref{eqn:Calpha}, \eqref{eqn:valpha},  and \eqref{eqn:Aalpha} respectively. 
Then, the Hessian of $ \|C_\alpha\|_F^2$, denoted by $H_\alpha:={\rm Hess}\big( \|C_\alpha\|_F^2\big)$ is expressed as follows
\begin{align}\label{eqn:hessian}
\frac{1}{4} H_\alpha =& MI+{\bf 1} {\bf1}\!^T -2M\, v(\alpha) \left(A_\alpha^{-1}{\bf 1} {\bf1}\!^T +{\bf 1} {\bf1}\!^T A_\alpha^{-1} \right) \nonumber \\
&+\left(3M^2\,v(\alpha)-c-2Ma\right)v(\alpha) A_\alpha^{-1} {\bf 1}{\bf1}\!^T A_\alpha^{-1}\nonumber \\
&-\left(M^2\,v(\alpha)-c-2Ma\right)A_\alpha^{-1},
\end{align}
where $I$ is the identity matrix, $c={\bf 1}\!^T \!\widetilde{C} {\bf1}$,  $\,a:={\bf 1}\!^T \alpha$.
\end{lem}
\noindent
The proof of Lemma \ref{lem:Hessian} is given in Section \ref{sec:proofs}.

\medskip
\noindent
Next, we show the convexity of $\|C_\alpha\|_F^2$. This is needed for the validity of optimization algorithms such as trust-region or gradient descent.
\begin{thm}\label{thm:Convex}
Suppose the empirical covariance matrix  $\widetilde{C}$ is of full rank, and the quantities $C_\alpha$ and  $v(\alpha)$ are as in \eqref{eqn:Calpha} and \eqref{eqn:valpha}. 
Then, the Frobenius norm squared $ \|C_\alpha\|_F^2$ is convex, i.e.,
\begin{equation}\label{eqn:laplacian}
\triangle  \|C_\alpha\|_F^2 \geq 0 \quad \forall \alpha.
\end{equation}
\end{thm}
\noindent
The proof of Theorem \ref{thm:Convex} is given in Section \ref{sec:proofs}.

\begin{figure*}
\begin{center}
\includegraphics[width=\linewidth]{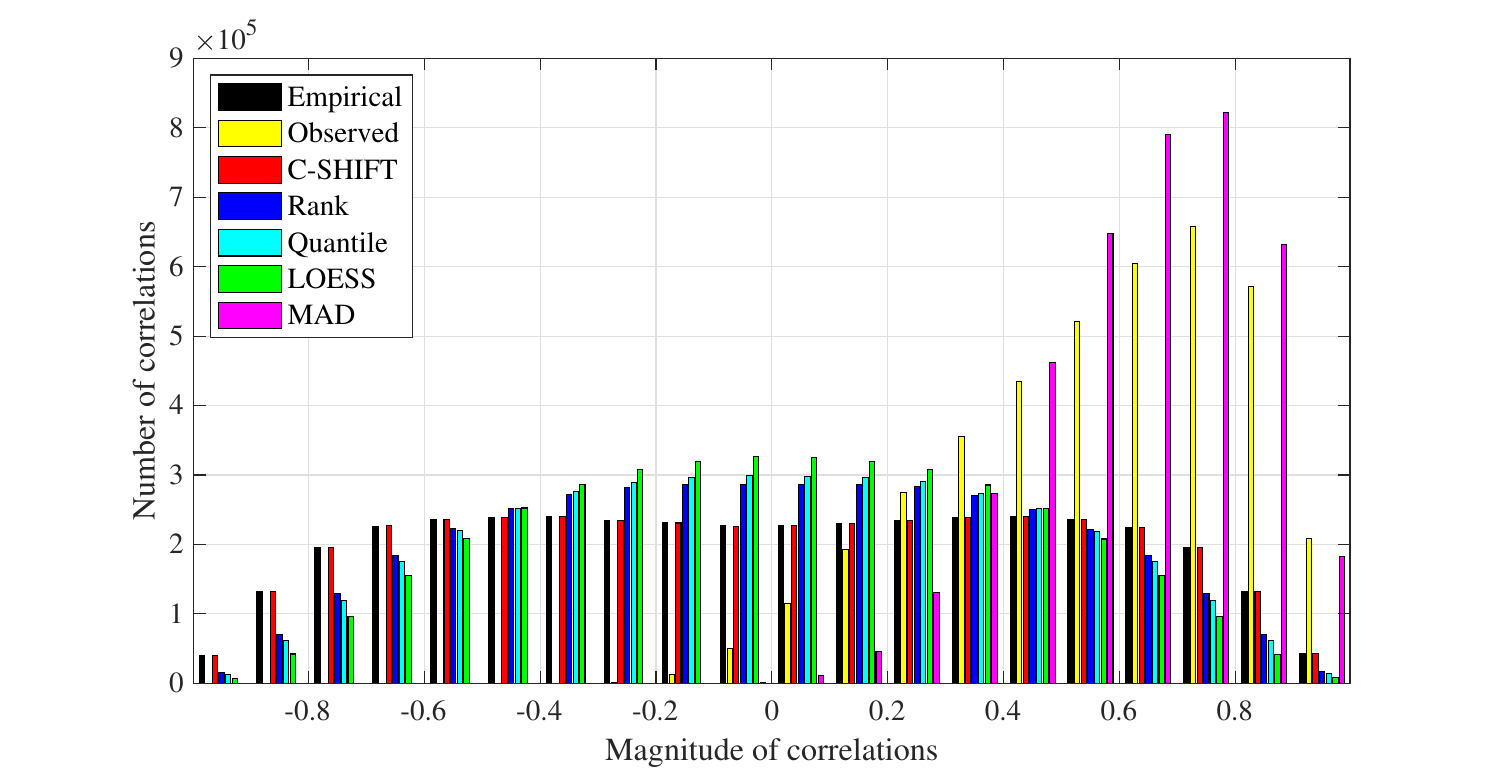}
\caption{A bar graph of correlations for the RCM dataset. On the x-axis we display the range of correlations, partitioned into intervals of length $0.1$. The height of each bar describes the number of correlations that belong to the corresponding interval. Bars of different colors correspond to different correlation matrices, indicated in the legend.}
\label{RCM_corr}
\end{center}
\end{figure*}
\begin{figure*}
\begin{center}
\includegraphics[width=\linewidth]{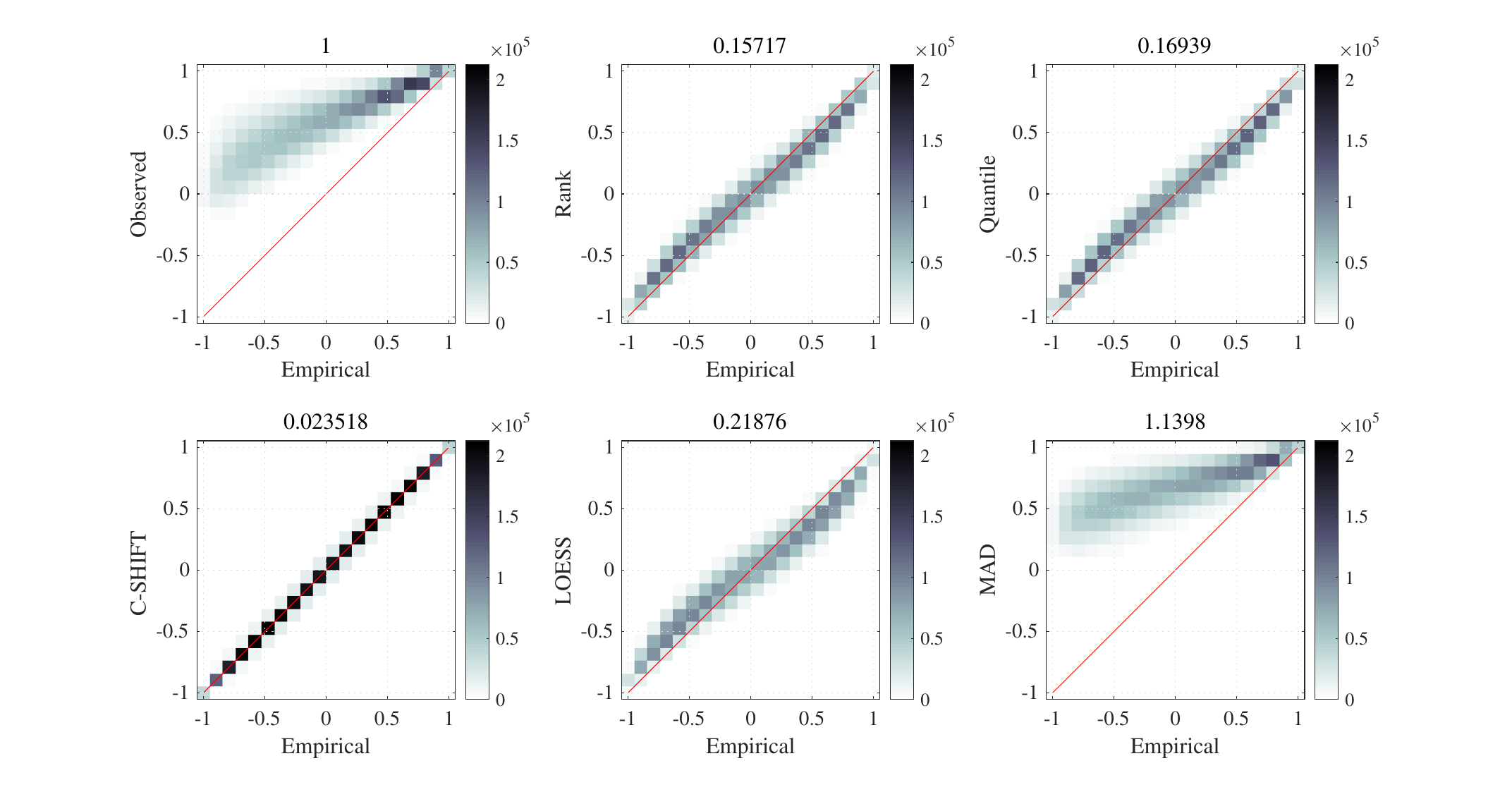}
\caption{The heat maps for the RCM dataset. Each heat map illustrates the transformation of the true empirical correlations $corr(Y_n,Y_m)$ (horizontal axis) after adding bias and applying the corresponding normalization method. In the top left plot the vertical axis represents the observed correlations $corr(\widetilde{Y}_n,\widetilde{Y}_m)$. In the remaining five heat maps, the vertical coordinates represent the correlations after normalization. Going clockwise, these five heat maps are Rank, Quantile, MAD, LOESS, and C-SHIFT. The darker the color, the higher the density. The number on top of each heat map indicates the relative leftover error after normalization. Smaller numbers indicate better recovery performance.}
\label{RCM_heat}
\end{center}
\end{figure*}

\begin{figure*}
\begin{center}
\includegraphics[width=\linewidth]{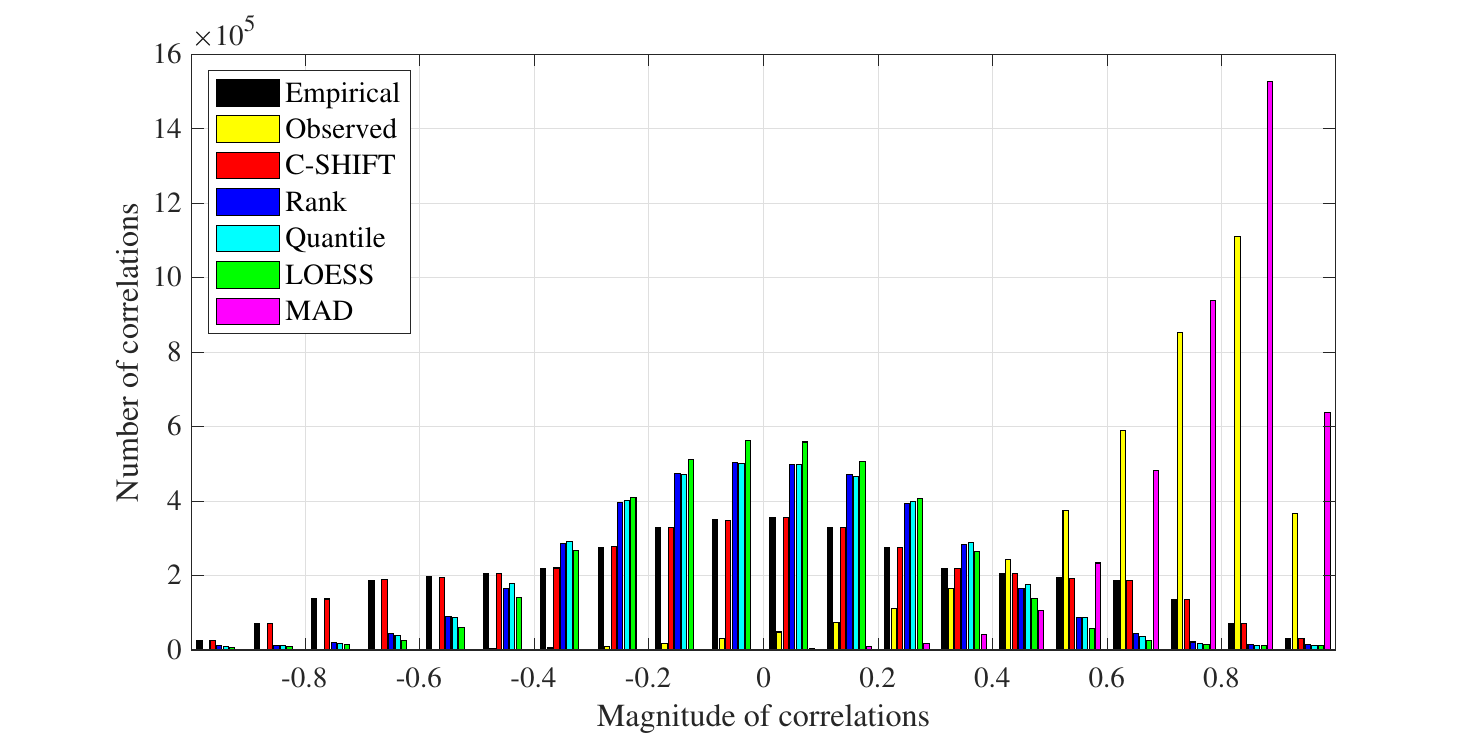}
\caption{A bar graph of correlations for the Cascade dataset. On the x-axis we display the range of correlations, partitioned into intervals of length $0.1$. The height of each bar describes the number of correlations that belong to the corresponding interval. Bars of different colors correspond to different correlation matrices, indicated in the legend.}
\label{Casc_corr}
\end{center}
\end{figure*}
\begin{figure*}
\begin{center}
\includegraphics[width=\linewidth]{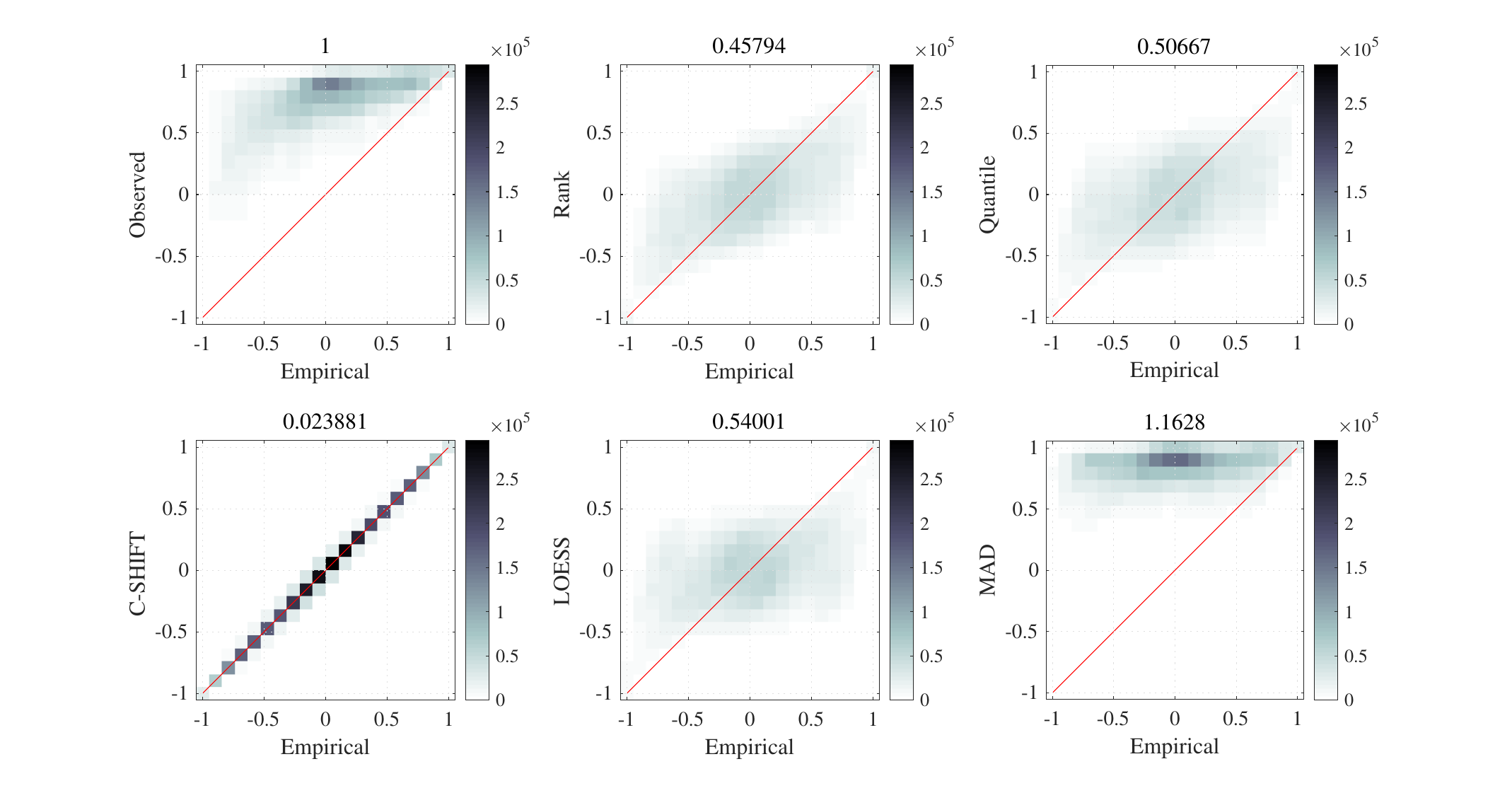}
\caption{The heat maps for the Cascade dataset. Each heat map illustrates the transformation of the true empirical correlations $corr(Y_n,Y_m)$ (horizontal axis) after adding bias and applying the corresponding normalization method. In the top left plot the vertical axis represents the observed correlations $corr(\widetilde{Y}_n,\widetilde{Y}_m)$. In the remaining five heat maps, the vertical coordinates represent the correlations after normalization. Going clockwise, these five heat maps are Rank, Quantile, MAD, LOESS, and C-SHIFT. The darker the color, the higher the density. The number on top of each heat map indicates the relative leftover error after normalization. Smaller numbers indicate better recovery performance.}
\label{Casc_heat}
\end{center}
\end{figure*}

\section{C-SHIFT algorithm and experiments}

In this section we provide the covariance shift (C-SHIFT) algorithm and evaluate its performance on synthetic datasets. Moreover, we compare the C-SHIFT algorithm with the well-known and frequently used normalization methods: Quantile, Rank, LOESS, and Median absolute deviation (MAD). Our empirical results demonstrate that the C-SHIFT algorithm outperforms other methods.

\subsection{C-SHIFT algorithm}\label{sec:alg}

The pseudocode for the C-SHIFT algorithm is given in Algorithm \ref{alg:C-SHIFT}.  Note that the algorithms takes into account both cases: when $\widetilde{C}$ has full rank and when $\widetilde{C}$ is rank deficient (i.e., $\widetilde{C}$ is positive semi-definite but not positive definite). When $\widetilde{C}$ is rank deficient the rank of $\,\widetilde{C}+ \alpha{\bf 1}\!^T+ {\bf 1} \alpha^T$ may exceed the rank $\widetilde{C}$ by no more than $2$, and therefore may also be rank deficient. Therefore, to make $\widetilde{C}$ a full rank we add to it a diagonal matrix ${\sf diag}(f)$, where $f$ is a vector of i.i.d. random variables from ${\sf Unif}[0,1]$. 

To find the optimal $\alpha^* = \arg \min_{\alpha} \|C_{\alpha}\|_F^2$, we use gradient and Hessian, provided in equations \eqref{eqn:gradient} and \eqref{eqn:hessian}, in the trust-region algorithm to minimize $\|C_\alpha\|_F^2$.

\begin{algorithm}[tb]
\caption{ C-SHIFT }
\label{alg:C-SHIFT}
\begin{algorithmic}
\STATE {\bf Input}: observed covariance matrix $\widetilde C$
\STATE {\bf Output}: recovered empirical covariance\\ matrix $C$
\IF {$\widetilde C$ is rank deficient}
\STATE $f \gets $ i.i.d. {\sf Unif}\,[0,1]
\STATE $\tilde C \gets \tilde C+{\sf diag}(f)$
\ENDIF
\STATE $v(\alpha)\gets \left[{\bf 1}\!^T \,\big(\widetilde{C}+ \alpha{\bf 1}\!^T+ {\bf 1} \alpha^T\big)^{-1} \, {\bf1}\right]^{-1}$
\STATE $C_\alpha \gets \widetilde{C}+ \alpha{\bf 1}\!^T+ {\bf 1} \alpha^T-v(\alpha){\bf 1\!1}\!^T$
\STATE $\alpha^*\gets \arg \min_{\alpha} \|C_{\alpha}\|_F^2$
\STATE $C\gets C_{\alpha^*}$
\IF {$\widetilde C$ is rank deficient}
\STATE $C \gets C-{\sf diag}(f)$
\ENDIF
\STATE {\bf return} $C$ 
\STATE{\quad} 
\end{algorithmic}
\end{algorithm}

\subsection{Numerical experiments}\label{sec:num}

In this section we conduct experiments on two synthetic datasets that we generate using random covariance method (RCM) and cascade method. We start by describing both methods. 

\subsubsection{Data generation}
\label{data_gen}
\paragraph{Random covariance method (RCM)} We generate a synthetic dataset with $M = 2000$ genes and $N = 50$ measurements (samples) using RCM. For that we first generate an auxiliary matrix $H \in \mathbb{R}^{M\times m}$ ($m=2$) whose entries are independent random variables, uniformly distributed over the interval $I=[-10,10]$. Next, we sample a diagonal matrix $D \in \mathbb{R}^{M\times M}$ with diagonal entries being i.i.d. exponential random variables with parameter $\lambda_D=30$. We let $\Sigma=HH^T+D$ be the  population (parameter) covariance matrix. Then we generate the true empirical logarithmic data $Y^{(i)}=\big(Y_n^{(i)}\big) \sim \mathcal{N}\big(0,\Sigma\big)$ for each $i=1,\hdots,N$. Finally, we set the observed logarithmic data be $\widetilde{Y}_n^{(i)}=Y_n^{(i)}+V^{(i)}$, where vector $V^{(i)}$ are $\mathcal{N}\big(-0.01,100\big)$ random variables. 

\paragraph{Cascade method} The {\it cascade} datasets were generated according by a directed acyclic weighted network $G=(V,E)$  aka directed acyclic graph (DAG). The graph was randomly generated via a recurrent {\it cascade} model. The parent-offspring relation is represented by the direction of edges $E=\{(u,v)\}$ of the graph $G$, i.e., $u$ is the parent vertex and $v$ is its offspring. For any vertex $v$ let $pa(v)$ be the set of its parents, $pa(v) = \{ u\in V: (u,v)\in E\}$.
Next, for each edge $(u,v) \in E$ an independent random weight $c_{uv}$ is assigned, with c.d.f. 
$$p\,U_{[a_-,b_-]}(x)+(1-p)\, U_{[a_+,b_+]}(x),$$
where the parameters $\,a_- < b_- \leq 0$, $\,0 \leq a_+ < b_+$, and $p \in (0,1)$ are fixed,
and $U_A(x)$ denotes the uniform c.d.f. on an interval $A$. We generated a random weighted DAG with the nodes $v \in V$ representing the genes. The random variables  $\{Y_v\}_{v\in V}$ representing the logarithmic gene expressions are generated as a noisy multiplicative cascade via the following structural linear recursive equations:
$$
Y_v = \sum_{u\in pa(v)} c_{uv} Y_{u} + \varepsilon_v, 
$$
where the recursion begins with $Y_0=y_0$, and proceeds from generation to generation. The noise
variables $(\varepsilon_v, v\in V)$ are i.i.d. $\mathcal{N}\big(0, \sigma^2\big)$, sampled independently from the random weights $c_{uv}$.
For simulation of $(Y_v, v\in V)$ the following values of parameters were chosen:
\begin{center}
\begin{tabular}{c|c|c|c|c|c|c|c|c}
 $p$ & $[a_-,b_-]$ & $[a_+,b_+]$ & $\sigma^2$ & $y_0$ & $|V|$\\
\hline
 $1/3$ & $[-1.2,-0.5]$ & $[0.5,1.3]$ & $1$ & $4.5$ & $2000$
\end{tabular}
\end{center}

\subsubsection{Simulation results}

We generate two datasets (RCM and Cascade) using the methods described in section \ref{data_gen}. Each date set consists of a matrix with the empirical data $\big(Y_n^{(i)}\big) \in \mathbb{R}^{M\times N}$ and a matrix with the observed data $\big(\widetilde{Y}_n^{(i)}\big)\in \mathbb{R}^{M\times N}$. In both, RCM and Cascade datasets, we let $M = 2000$ genes and $N=50$ measurements (samples). 
For each dataset, we normalize the covariance matrix $\widetilde{C}$, obtained from the observed data, by using C-SHIFT, Rank, Quantile, LOESS, and MAD methods. We compare the performance of the algorithms using the results presented in Figures \ref{RCM_corr}-\ref{Casc_heat}. 

In Figures \ref{RCM_corr} and \ref{Casc_corr} we depict the bar graphs of correlations for RCM and Cascade datasets, respectively. 
As we can see in both datasets, the correlations of the observed data (yellow) are shifted away from $-1$ so that there are no large magnitude negative correlations. The aim of the normalization algorithms is to shift the correlations back into correct positions, i.e., ideally, the correlations of the normalized data should match the empirical correlations. Note that for both datasets, the C-SHIFT method correctly recovers the number of correlations in each interval: the red bars almost perfectly match the black bars. In contrast, other normalization methods could not recover the correct numbers of correlations, especially for the correlations of larger magnitudes. Specifically, Rank, Quantile and LOESS normalization techniques tend to shift correlations mostly to the center of the bar plot, each forming a bell shape. Predictably, the MAD method has the worst performance in correlation recovery. Finally, among the other three normalization techniques (Quantile, Rank, and LOESS), the latter method has the poorest  performance.  

Figures \ref{RCM_heat} and \ref{Casc_heat} contain six heat maps each, for RCM and Cascade datasets, respectively. Each heat map illustrates the transformation of the true empirical correlations $corr(Y_n,Y_m)$ (horizontal axis) after adding bias and applying the corresponding normalization method. We consider 2,001,000 correlations corresponding to all pairs of genes. For each point, representing a pair of genes $(n,m)$, the horizontal coordinate equals the true empirical correlation $corr(Y_n,Y_m)$ in all six plots. The vertical coordinate in the top left heat map is the correlation in the observed data, $corr(\widetilde{Y}_n,\widetilde{Y}_m)$. Importantly, it shows the shift of correlations rightward in the observed data. In the remaining five heat maps, the vertical coordinates represent the correlations after normalization. Going clockwise, these five heat maps are Rank, Quantile, MAD, LOESS, and C-SHIFT. The darker the color, the higher the density. Notice that the heat map for C-SHIFT is almost perfectly diagonal, which demonstrates how well C-SHIFT recovers the correlations. Thus, in addition to correctly recovering the right numbers of correlations in each interval (which was demonstrated in Figures \ref{RCM_corr} and \ref{Casc_corr}), the proposed C-SHIFT algorithm also returns (shifts back) the correlations to the correct margins. Hence, the heat map is a diagonal line. The number on top of each heat map indicates the relative leftover error after normalization, i.e., the $\ell^2$-norm of the vector of differences between the horizontal and vertical coordinates, scaled by the Frobenius norm of the difference between the empirical and the observed correlation matrices. Thus, the left top heat map is assigned the value $1$, and for each normalization method, the smaller the number the better it recovers the empirical correlation matrix. Any such number smaller than one is an improvement. The number for C-SHIFT is by far the smallest in each dataset ($0.023518$ and $0.023881$), while in the case of MAD normalization, the corresponding number even exceeds $1$.

\section{Evaluation of C-SHIFT algorithm on real data}\label{sec:real}

In this section we apply C-SHIFT algorithm to real biological data, and compare the resulting correlations to 
the correlations obtained by normalizing the same data with Rank, Quantile, and LOESS.
In the analysis, we used scaled $\ell^1$-norm to measure the distance between correlation matrices.
Specifically, for two correlation matrices, $R=\big(r_{i,j}\big)$ and $R'=\big(r'_{i,j}\big)$, the norm
\begin{equation}\label{eqn:CorrNorm}
d(R,R')={1 \over M(M-1)}\sum\limits_{1\leq i<j\leq M} \!\!\!\!\!|r_{i,j}-r'_{i,j}|
\end{equation}
measures the distance between $R$ and $R'$ on the scale from $0$ to $1$.
We considered the following microarray datasets from GEO depository.

\begin{table*}[htbp]\caption{Distance \eqref{eqn:CorrNorm} between pairs of correlation matrices recovered by normalization methods
in the six datasets.  
 }
\centering 
\begin{tabular}{|l|c|c|c|c|c|c}
\toprule 
 Dataset & Rank vs. C-SHIFT & Quantile vs. C-SHIFT & LOESS vs. C-SHIFT \\
\hline
 {\bf GSE7803[Carcinoma]} & $0.02824$ & $0.017624$ & $0.017557$ \\
 {\bf GSE7803[Normal]} & $0.03425$ & $0.023438$ & $0.025578$ \\
 {\bf GSE152738} & $0.019763$ & $0.014113$ & $0.015207$ \\
 {\bf GSE86858} & $0.096963$ & $0.095617$ & $0.10688$ \\
 {\bf GSE59412[TIVE]} & $0.046627$ & $0.041095$ & $0.043859$ \\
 {\bf GSE59412[BJAB]} & $0.041824$ & $0.038841$ & $0.04086$ \\
 \bottomrule
\end{tabular}
\label{tab:dist}
\end{table*}

\medskip
\noindent
Two datasets come from GSE7803 in GEO depository \cite{zhai2007gene}. 
Dataset {\bf GSE7803[Carcinoma]} looks into $21$ samples of invasive squamous cell carcinomas 
and $4,152$ genes. 
Dataset {\bf GSE7803[Normal]} has $10$ normal cervical samples and $4,709$ genes.

\medskip
\noindent
Dataset {\bf GSE152738}  from GEO depository \cite{nishimura2021chitinase} consists of $58$ liver specimens from adult liver donors and $12,164$ genes. 

\medskip
\noindent
Dataset {\bf GSE86858} obtained in \cite{kozuka2017marked} has $15,312$ genes and $8$ samples from obese diabetic mice treated with $\gamma$-oryzanol-encapsulated nanoparticles, of which, $4$ were taken from liver and $4$ from hypothalamus.

\medskip
\noindent
Two datasets come from GSE59412 \cite{yang2014systems}, where it was discovered that the ectopic expression of miR-K12-11 differentially affected gene expression in BJAB cells of lymphoid origin and TIVE cells of endothelial origin. 
Dataset {\bf GSE59412[TIVE]} consists of $8$ samples of TIVE cells and $16,700$ genes.
Dataset {\bf GSE59412[BJAB]} consists of $24$ samples of BJAB cells and $19,296$ genes.

\medskip
\noindent
All six datasets considered 
were not normalized prior to the analysis. 
Affy R package and MAS-5 method \cite{arteaga2008overview,gautier2004affy,irizarry2003r,pepper2007utility} was used 
for reading and preliminary data analysis at the probe-level of affymetrix CEL files. 
We calculated Abesnt/Present Call for each probe set and subselected only the genes that are expressed in all samples.

\medskip
\noindent
Next, we summarize our observations.
First, we notice that in all real and synthetic datasets considered in this analysis, the correlations produced by Rank, Quantile, and LOESS 
are close to each other. In the synthetic data, where the desired true empirical correlations are known, 
one easily encounters a situation where under a strong bias the correlations produced by C-SHIFT are significantly
different from the correlations produced by Rank, Quantile, or LOESS. See Fig.~\ref{Casc_heat}.

Recall that $d(R,R')$ defined in \eqref{eqn:CorrNorm} measures the distance between correlation matrices on the scale from $0$ to $1$.
In the six real datasets considered in this work, we notice that the distances between the correlations obtained from C-SHIFT 
and either one of the three normalization methods used in comparison (Rank, Quantile, and LOESS) range between $0.01$ and $0.1$.
See Table \ref{tab:dist} and Figures \ref{GSE7803_tumor}, \ref{GSE7803_normal}, \ref{GSE152738}, \ref{GSE86858}, \ref{GSE59412_TIVE}, and \ref{GSE59412_BJAB}.
In five out of six datasets, the distance between C-SHIFT and any of the three normalization approaches does not exceed $0.05$.  
A small but sizable mismatch of $\approx 0.1$ between C-SHIFT and each of the three normalization methods is observed in dataset GSE86858.

\begin{figure*}
\begin{center}
\includegraphics[width=0.9\linewidth]{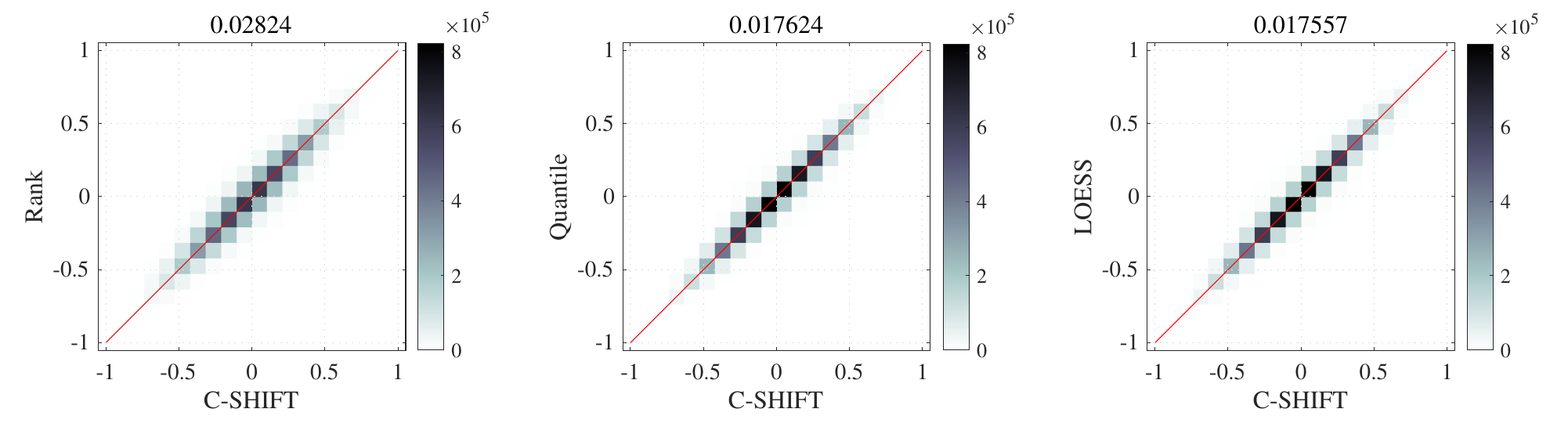}
\includegraphics[width=0.9\linewidth]{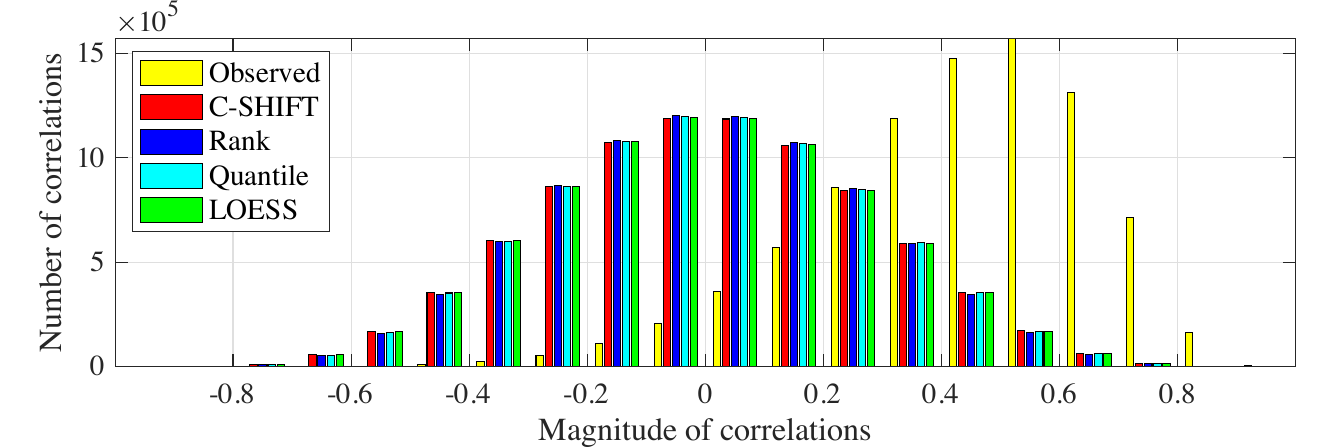}
\caption{Top row: heat maps for {\bf GSE7803[Carcinoma]} dataset that compare the correlations obtained by Rank, Quantile, and LOESS to C-SHIFT.
The number on top of each heat map represents the distance between correlation matrices as defined in \eqref{eqn:CorrNorm}.
Bottom row: numbers of correlations for {\bf GSE7803[Carcinoma]} dataset. Different colors correspond to different correlation matrices, indicated in the legend.
Horizontal axis is partitioned into intervals of length $0.1$.}
\label{GSE7803_tumor}
\end{center}
\end{figure*}
\begin{figure*}
\begin{center}
\includegraphics[width=0.9\linewidth]{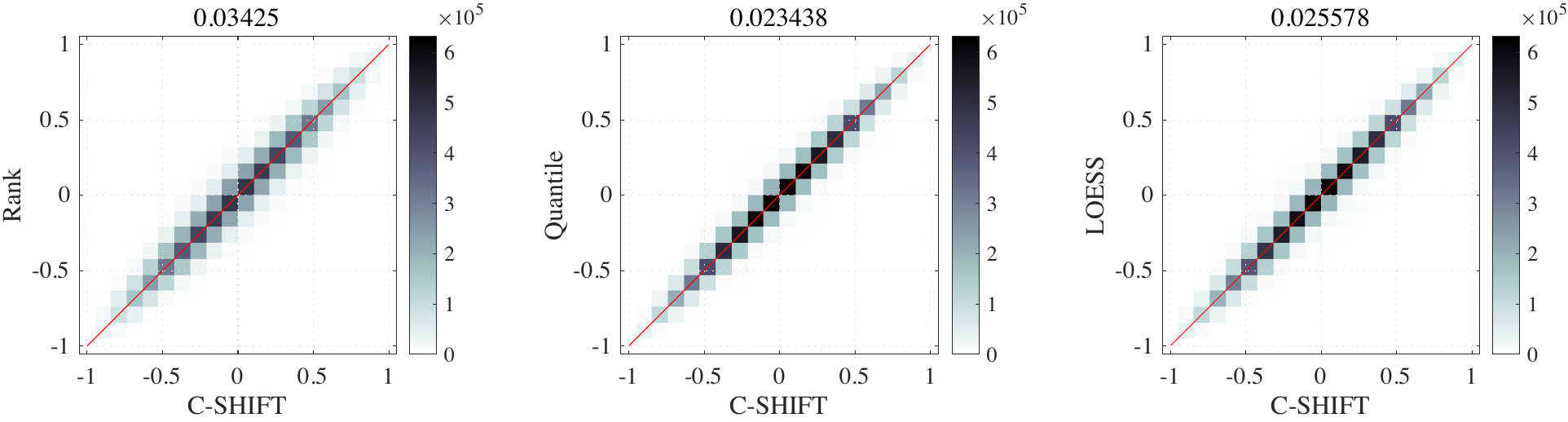}
\includegraphics[width=0.9\linewidth]{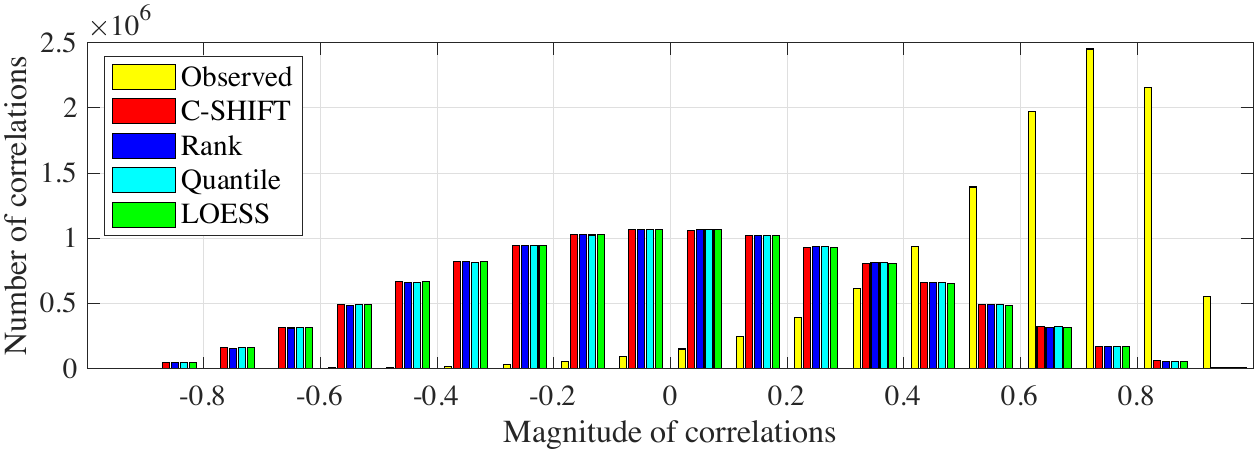}
\caption{Top row: heat maps for {\bf GSE7803[Normal]} dataset that compare the correlations obtained by Rank, Quantile, and LOESS to C-SHIFT.
The number on top of each heat map represents the distance between correlation matrices as defined in \eqref{eqn:CorrNorm}.
Bottom row: numbers of correlations for {\bf GSE7803[Normal]} dataset. Different colors correspond to different correlation matrices, indicated in the legend.
Horizontal axis is partitioned into intervals of length $0.1$.}
\label{GSE7803_normal}
\end{center}
\end{figure*}

\begin{figure*}
\begin{center}
\includegraphics[width=0.9\linewidth]{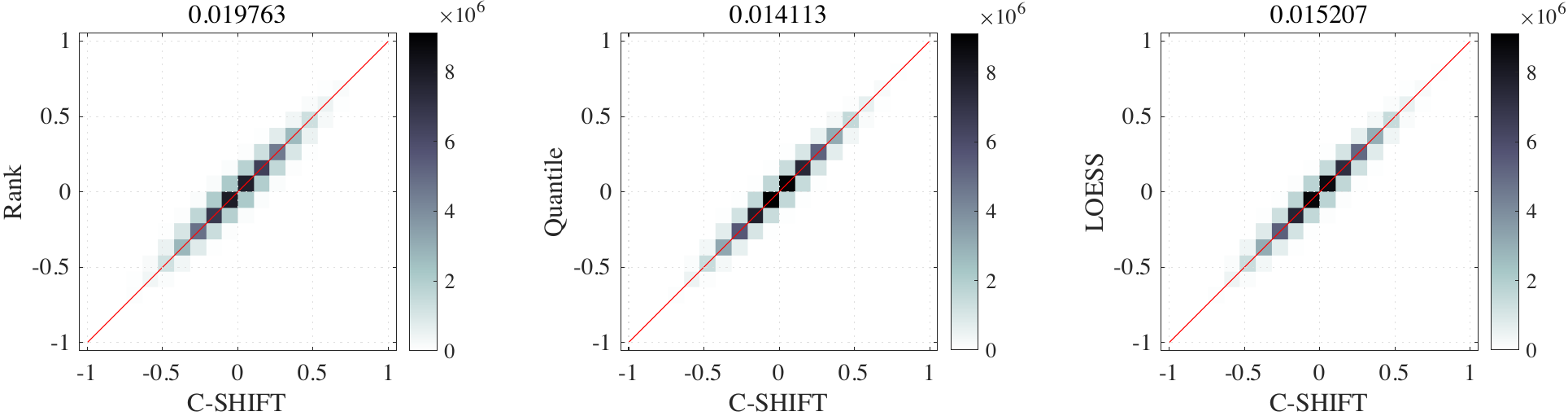}
\includegraphics[width=0.9\linewidth]{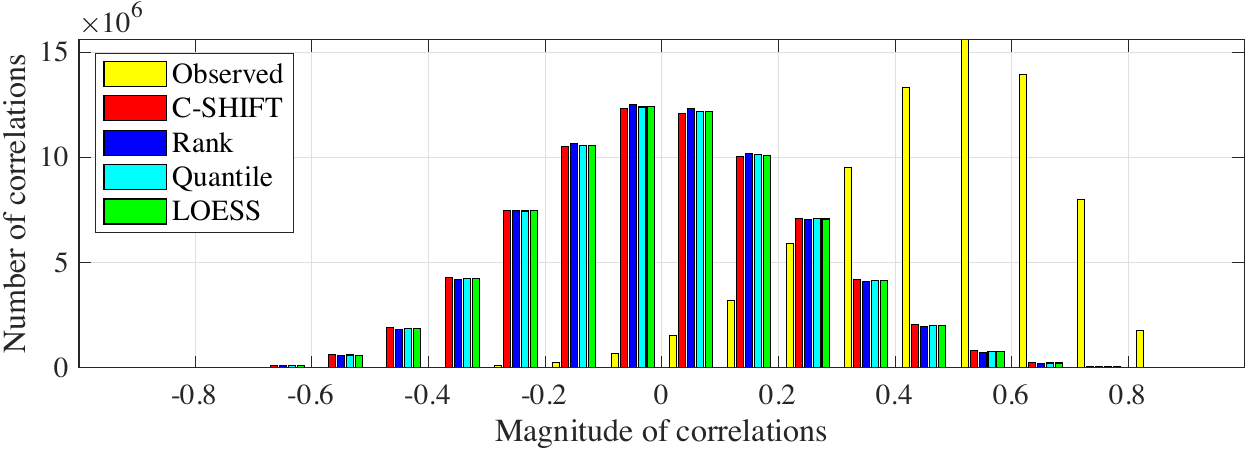}
\caption{Top row: heat maps for {\bf GSE152738} dataset that compare the correlations obtained by Rank, Quantile, and LOESS to C-SHIFT.
The number on top of each heat map represents the distance between correlation matrices as defined in \eqref{eqn:CorrNorm}.
Bottom row: numbers of correlations for {\bf GSE152738} dataset. Different colors correspond to different correlation matrices, indicated in the legend.
Horizontal axis is partitioned into intervals of length $0.1$.}
\label{GSE152738}
\end{center}
\end{figure*}
\begin{figure*}
\begin{center}
\includegraphics[width=0.9\linewidth]{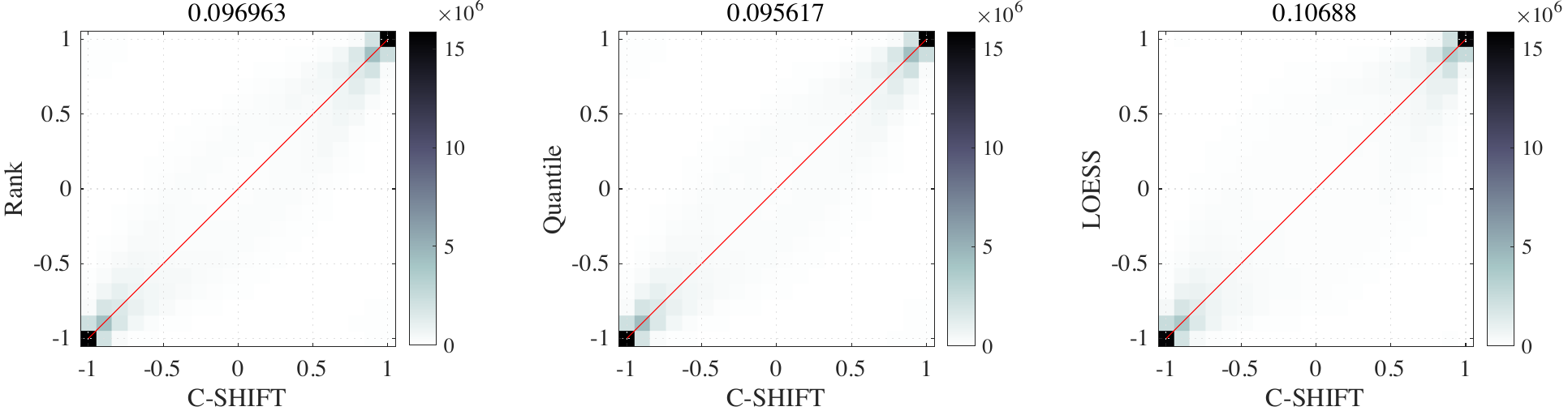}
\includegraphics[width=0.9\linewidth]{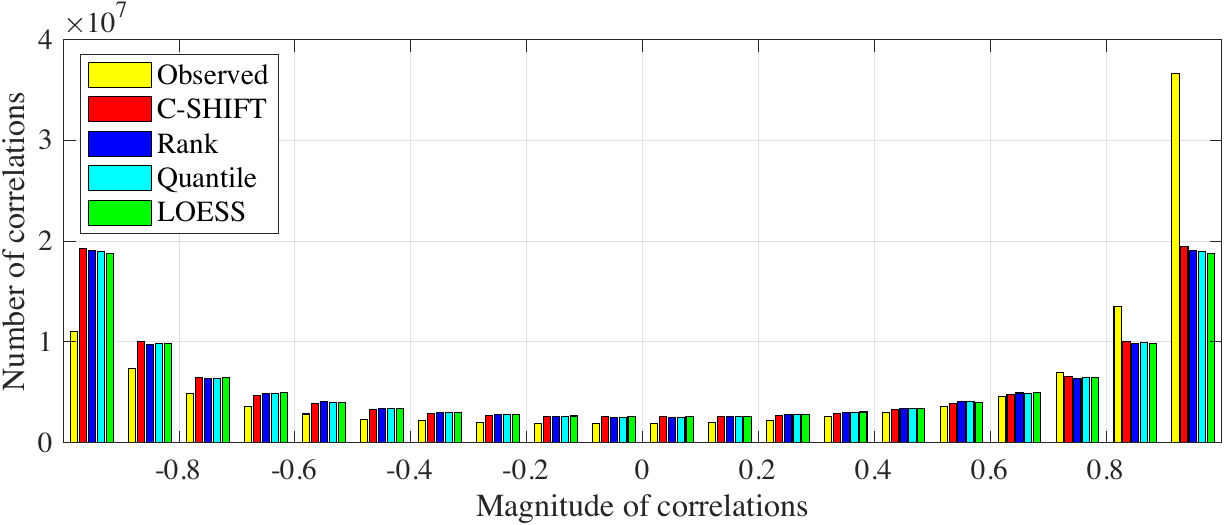}
\caption{Top row: heat maps for {\bf GSE86858} dataset that compare the correlations obtained by Rank, Quantile, and LOESS to C-SHIFT.
The number on top of each heat map represents the distance between correlation matrices as defined in \eqref{eqn:CorrNorm}.
Bottom row: numbers of correlations for {\bf GSE86858} dataset. Different colors correspond to different correlation matrices, indicated in the legend.
Horizontal axis is partitioned into intervals of length $0.1$.}
\label{GSE86858}
\end{center}
\end{figure*}

\begin{figure*}
\begin{center}
\includegraphics[width=0.9\linewidth]{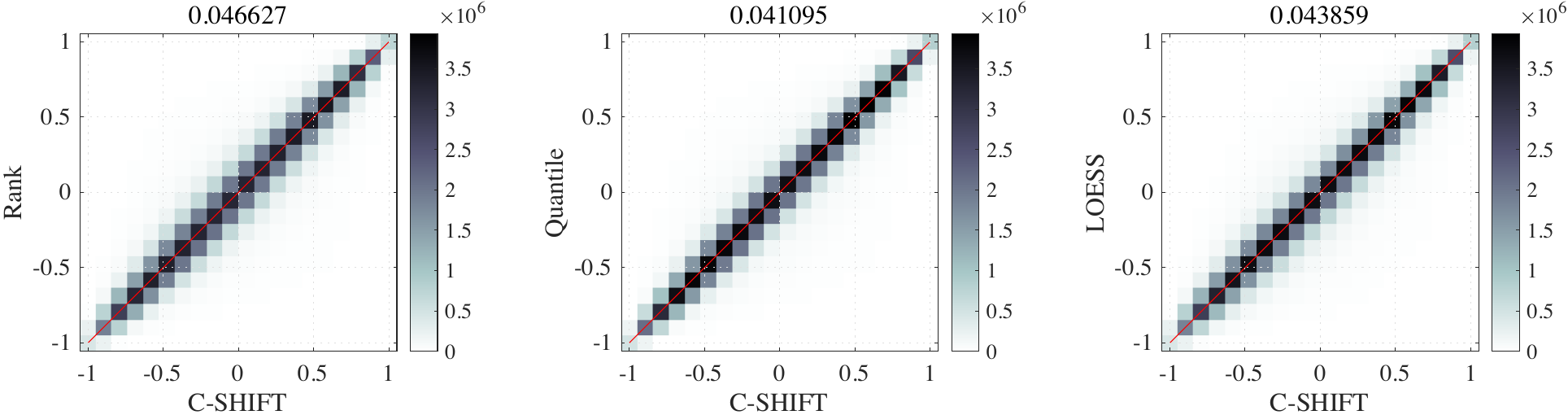}
\includegraphics[width=0.9\linewidth]{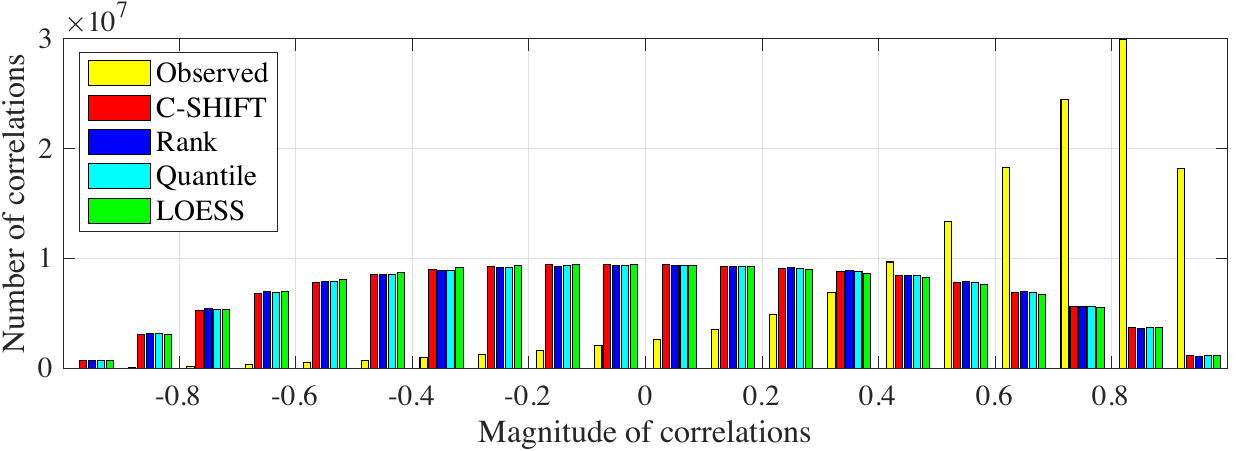}
\caption{Top row: heat maps for {\bf GSE59412[TIVE]} dataset that compare the correlations obtained by Rank, Quantile, and LOESS to C-SHIFT.
The number on top of each heat map represents the distance between correlation matrices as defined in \eqref{eqn:CorrNorm}.
Bottom row: numbers of correlations for {\bf GSE59412[TIVE]} dataset. Different colors correspond to different correlation matrices, indicated in the legend.
Horizontal axis is partitioned into intervals of length $0.1$.}
\label{GSE59412_TIVE}
\end{center}
\end{figure*}
\begin{figure*}
\begin{center}
\includegraphics[width=0.9\linewidth]{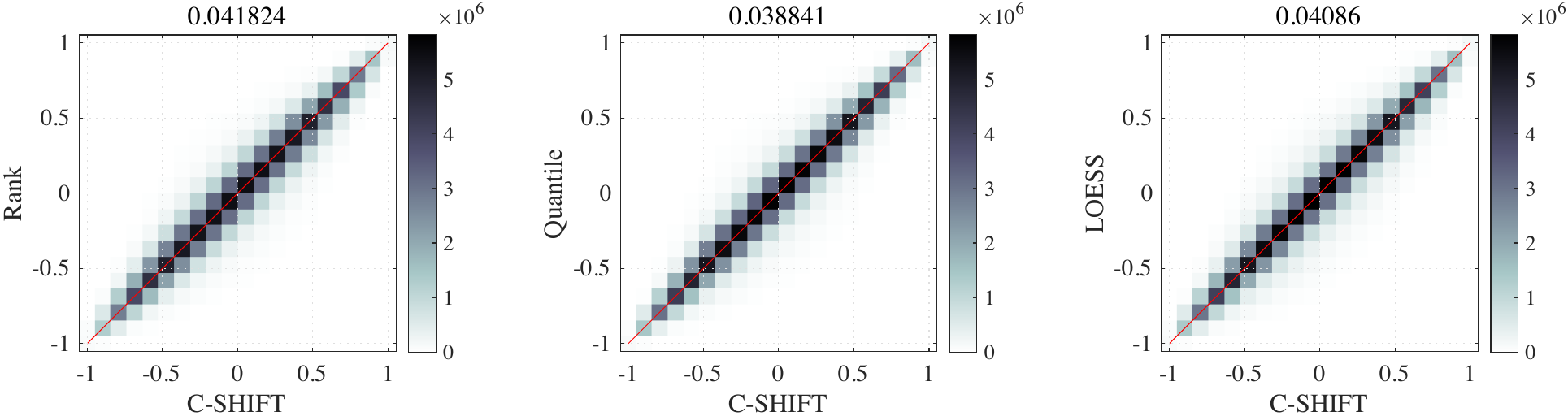}
\includegraphics[width=0.9\linewidth]{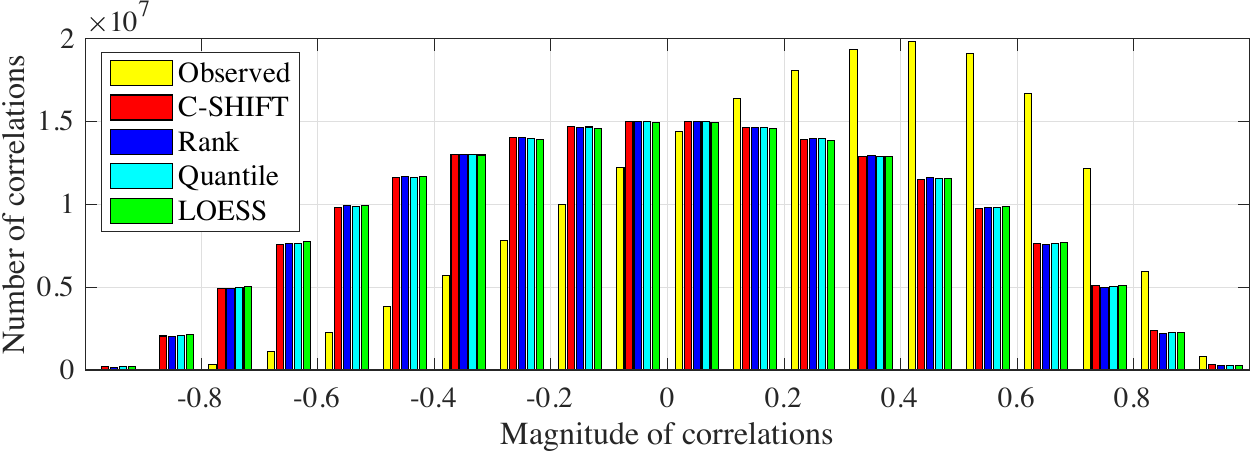}
\caption{Top row: heat maps for {\bf GSE59412[BJAB]} dataset that compare the correlations obtained by Rank, Quantile, and LOESS to C-SHIFT.
The number on top of each heat map represents the distance between correlation matrices as defined in \eqref{eqn:CorrNorm}.
Bottom row: numbers of correlations for {\bf GSE59412[BJAB]} dataset. Different colors correspond to different correlation matrices, indicated in the legend.
Horizontal axis is partitioned into intervals of length $0.1$.}
\label{GSE59412_BJAB}
\end{center}
\end{figure*}

\section{Discussion}\label{sec:discuss}
In systems biology, the gene co-expression networks (GCN) are reconstructed from the correlations between the genes. GCN recovery relies on removing the bias with a normalization method, and thus improving the estimation of correlations between the pairs of genes. However, the standard normalization techniques such as Rank, Quantile, LOESS, and MAD are known to be insufficient at recovering true empirical correlations while the C-SHIFT algorithm is specifically designed to recover the true empirical correlations. The multiple experiments with synthetic datasets demonstrate the algorithm's superior performance at recovering true empirical correlations
in comparison to the standard normalization techniques. 

Also, we observed that the correlations recovered by C-SHIFT, Rank, Quantile, and LOESS would essentially match in five out of six real datasets considered in this paper.
One dataset demonstrated small but sizable difference hinting at a greater variance of the bias.

Importantly, we notice that the C-SHIFT algorithm corrects the positive shift of covariances (and correlations) observed when $\hat{\omega}=\widehat{Var}(V)$ is larger than $\hat{a}_n=-\widehat{Cov}(Y_n,V)$ ($n=1,\hdots,M$) in \eqref{eqn:covLogEmpirical}. Hence, the independence of $V$ from $Y_n$ assumption can be replaced with 
a weaker assumption stating that $Cov(Y_n,V)\ll Var(V)$. This will be explored in a follow-up publication.

An alternative version of the C-SHIFT algorithm is based on trace minimization approach instead of energy minimization.
In this alternative C-SHIFT algorithm, the positive semi-definite matrix $C_{\alpha^*}$ with
$$\alpha^*=\text{argmin} \,\Tr\big(C_\alpha\big)$$
is used to approximate the true empirical covariance matrix $C$.
The analogs of Lemmas \ref{lem:alpha} and \ref{lem:Hessian} and the convexity result in Theorem \ref{thm:Convex} are also established for $\Tr\big(C_\alpha\big)$ in the trace minimization approach.
See \cite{logan2020thesis}.
Empirically it appears that this alternative approach produces the same $\alpha^*$ as the original C-SHIFT algorithm based on energy minimization as presented in this paper, and therefore it recovers the empirical covariance $C$ with the same accuracy.
Thus, the alternative, trace minimizing C-SHIFT algorithm can be used instead of Algorithm \ref{alg:C-SHIFT}.
This approach will be analyzed in a follow-up paper.

Finally, the C-SHIFT algorithm was deposited on GitHub at
\hyperref[https://github.com/evcphd/C-SHIFT]{https://github.com/evcphd/C-SHIFT}

\section{Proofs}\label{sec:proofs}

\begin{proof}[Proof of Proposition~\ref{prop:1C1}]
Observe that
$$x^T \big(\mathcal{M} -v\, {\bf 1\!1}\!^T \big)x= x^T \mathcal{M}x -v\left(\sum x_i\right)^2 \geq 0$$
for all $x \in \mathbb{R}^M$ if and only if $v \leq v^*$, where $v^*$ minimizes $x^T \mathcal{M}x$ under the condition $\sum x_i=Const$.
Next, applying the Lagrange multipliers method, we obtain $2\mathcal{M}x=\lambda {\bf 1}$, and therefore, 
$$v^*=\frac{x^T \mathcal{M}x}{(\sum x_i)^2}=\frac{\frac{\lambda}{2} x^T {\bf 1}}{(\sum x_i)^2}=\frac{\lambda/2}{{\bf 1}^T x}=\frac{1}{{\bf 1}\!^T \,\mathcal{M}^{-1} \, {\bf1}}$$
as $x=\frac{\lambda}{2}\mathcal{M}^{-1}{\bf 1}$.
\end{proof}

\begin{proof}[Proof of Lemma \ref{lem:alpha}]
By \eqref{eqn:Calpha}, we have
\begin{align}\label{eqn:CalphaFrobenius}
\|C_\alpha\|_F^2 =& \|\widetilde{C}\|_F^2+2M\!\sum\limits_{i=1}^M \alpha_i^2+M^2\,v^2(\alpha)+4\left({\bf 1}\!^T \!\widetilde{C}\alpha \right) \nonumber \\
&+ 2a^2  - 2c\,v(\alpha) -4Ma\,v(\alpha)
\end{align}
Notice that 
\begin{align}\label{eqn:daInverse}
&\frac{\partial}{\partial \alpha_i}A_\alpha=\bar{e}_i{\bf 1}^T+{\bf 1}\bar{e}_i^T \quad 
\text{ and } \nonumber \\
&\frac{\partial}{\partial \alpha_i}A_\alpha^{-1}=-A_\alpha^{-1}\,\big(\bar{e}_i{\bf 1}^T+{\bf 1}\bar{e}_i^T \big)A_\alpha^{-1},
\end{align}
where $\bar{e}_i$ is the $i$-th coordinate vector. Therefore, we have
\begin{align}\label{eqn:dava}
\frac{\partial}{\partial \alpha_i}  v(\alpha) &= v^2(\alpha){\bf 1}\!^T A_\alpha^{-1} \, \big(\bar{e}_i{\bf 1}^T+{\bf 1}\bar{e}_i^T \big)A_\alpha^{-1} \, {\bf1} \nonumber \\
&=2v(\alpha)\,{\bf 1}\!^T A_\alpha^{-1} \bar{e}_i
\end{align}
implying
\begin{equation}\label{eqn:va_gradient}
\nabla v(\alpha) =2v(\alpha) A_\alpha^{-1} {\bf 1}.
\end{equation}

\medskip
\noindent
Next, the gradient $\nabla \|C_\alpha\|_F^2$ in \eqref{eqn:gradient} is found via the equations \eqref{eqn:CalphaFrobenius} and \eqref{eqn:va_gradient}.
\end{proof}

\begin{proof}[Proof of Lemma \ref{lem:Hessian}]
By \eqref{eqn:gradient}, we have 
\begin{align}\label{eqn:HessianParts}
\frac{1}{4} H_\alpha &= \frac{1}{4} \nabla \left(\nabla  \|C_\alpha\|_F^2\right)^T \nonumber \\ 
&=M\nabla \alpha^T +\nabla {\bf 1}\!^T \big(a-M\,v(\alpha)\big) \\
&+ \Big(\nabla\big(M^2\,v^2(\alpha)-c\,v(\alpha)-2Ma\,v(\alpha)\big)\Big) {\bf1}\!^TA_\alpha^{-1}  \nonumber \\
&+\big(M^2\,v^2(\alpha)-c\,v(\alpha)-2Ma\,v(\alpha)\big) \nabla {\bf1}\!^T A_\alpha^{-1}, \nonumber 
\end{align}
where $\nabla=\left(\frac{\partial}{\partial \alpha_1},\hdots,\frac{\partial}{\partial \alpha_M} \right)^T$ was used as the column vector of the partial derivative operators. 
The summation parts in \eqref{eqn:HessianParts} are calculated as follows.
First, 
\begin{equation}\label{eqn:HessianPart1}
M\nabla \alpha^T =MI.
\end{equation}
Next, \eqref{eqn:va_gradient} implies
\begin{align}\label{eqn:HessianPart2}
&\nabla\big(M^2\,v^2(\alpha)-c\,v(\alpha)-2Ma\,v(\alpha)\big) \\
&~= 2\big(2M^2\,v(\alpha)-c-2Ma\big)v(\alpha)A_\alpha^{-1} {\bf1}  \nonumber
-2Mv(\alpha){\bf1}.
\end{align}
Equation \eqref{eqn:daInverse} yields
\begin{align}\label{eqn:HessianPart3}
\nabla {\bf1}\!^T A_\alpha^{-1} &= \sum\limits_{i=1}^M \bar{e}_i {\bf 1}\!^T \frac{\partial}{\partial \alpha_i}A_\alpha^{-1} \nonumber \\
&=- \sum\limits_{i=1}^M \bar{e}_i {\bf 1}\!^T A_\alpha^{-1}\,\big(\bar{e}_i{\bf 1}^T+{\bf 1}\bar{e}_i^T \big)A_\alpha^{-1} \nonumber \\
&=- \sum\limits_{i=1}^M \big(\bar{e}_i^T A_\alpha^{-1} {\bf 1}\big)\bar{e}_i  {\bf 1}^T A_\alpha^{-1} \nonumber \\
&~\quad-\big({\bf 1}\!^T A_\alpha^{-1}{\bf 1}\big)\sum\limits_{i=1}^M \bar{e}_i \bar{e}_i^T A_\alpha^{-1} \nonumber \\
&=-A_\alpha^{-1} {\bf 1}{\bf1}\!^T A_\alpha^{-1} -\big({\bf1}\!^T A_\alpha^{-1} {\bf 1}\big) A_\alpha^{-1} \nonumber \\
&=-A_\alpha^{-1} {\bf 1}{\bf1}\!^T A_\alpha^{-1} - \frac{1}{v(\alpha)} A_\alpha^{-1}.
\end{align}
Finally, \eqref{eqn:va_gradient} is used to derive
\begin{align}\label{eqn:HessianPart4}
\nabla {\bf 1}\!^T \big(a-M\,v(\alpha)\big) &={\bf 1}{\bf1}\!^T - 2M\, v(\alpha) A_\alpha^{-1}{\bf 1} {\bf1}\!^T.
\end{align}
Combining together equations \eqref{eqn:HessianPart1}-\eqref{eqn:HessianPart4} and substituting them into \eqref{eqn:HessianParts} we obtain \eqref{eqn:hessian}. 
\end{proof}

\begin{proof}[Proof of Theorem \ref{thm:Convex}]
We will use the notations from Lemmas \ref{lem:alpha} and \ref{lem:Hessian}
such as $c:={\bf 1}\!^T \!\widetilde{C} {\bf1}$ and $a:=\!\!\sum\limits_{i=1}^M \alpha_i$.
Without loss of generality we consider $\alpha$ on the hyperplane $a=0$.

\medskip
\noindent
Here, $A_\alpha =\widetilde{C}+ \alpha{\bf 1}\!^T+ {\bf 1} \alpha^T$ is a positive definite symmetric matrix with eigenvalues 
$\,\lambda_1 \geq \hdots \geq \lambda_M>0$
counted with respect to algebraic multiplicity, and let $\{v_i\}_{i=1,\hdots,M}$ be the corresponding orthonormal basis of eigenvectors.

\medskip
\noindent
As $\, \triangle  \|C_\alpha\|_F^2 =  \Tr\big(H_\alpha\big)$, equation \eqref{eqn:hessian} implies
\begin{align}\label{eqn:qtLaplacian}
\frac{1}{4} \triangle  \|C_\alpha\|_F^2 
&= M^2\Big(\!1-v(\alpha)\Tr\big(A_\alpha^{-1}\big)\Big)\! \nonumber \\
&\quad + \!c\Big(\!\Tr\big(A_\alpha^{-1}\big)- v(\alpha) {\bf 1}\!^T\! A_\alpha^{-2}{\bf1}\Big)\! \nonumber \\
&\quad +\!3M\Big(\!Mv^2(\alpha){\bf 1}\!^T \!A_\alpha^{-2}{\bf1}-1\Big).
\end{align}
The Laplacian in \eqref{eqn:qtLaplacian} is shown to be strictly positive in the following three steps.
First, by the Cauchy-Bunyakovsky-Schwarz 
inequality, we have
\begin{align}\label{eqn:First}
&Mv^2(\alpha){\bf 1}\!^T \!A_\alpha^{-2}{\bf1}-1 \nonumber \\
&\quad= v^2(\alpha)\Big(\big\|{\bf 1}\big\|_2^2 \big\|A_\alpha^{-1}{\bf1}\big\|_2^2-\big({\bf 1}\!^T \!A_\alpha^{-1}{\bf1}\big)^2\Big) \geq 0.
\end{align}

\noindent
Next, observe that $\,Mx+(1-x)^2 \geq 1\,$ for $M \geq 2$, 
and all $x \in [0,1]$.
Thus, for a given probability mass function $\{p_k\}_{k=1,\hdots,M}$ such that $p_k<1$ for all $k$,
and a given index $i \in \{1,\hdots,M\}$, Jensen's inequality implies
\begin{align}\label{eqn:Second01}
Mp_i&+\left(\sum\limits_{j:j\not=i} \lambda_j^{-1}p_j \right)\left(\sum\limits_{j:j\not=i} \lambda_j p_j \right)\nonumber \\
&=Mp_i+(1-p_i)^2\left(\sum\limits_{j:j\not=i} \lambda_j^{-1}q_j \right)\left(\sum\limits_{j:j\not=i} \lambda_j q_j\right)\nonumber \\
&\geq Mp_i+(1-p_i)^2 ~\geq 1
\end{align}
where we let $q_j={p_j \over 1-p_i}$ for all $j\not=i$.
Summing over all $i$ in \eqref{eqn:Second01}, we obtain,
\begin{align}\label{eqn:Second02}
\sum\limits_i \lambda_i^{-1}p_i+&{1 \over M}\sum\limits_i \lambda_i^{-1} \left(\sum\limits_{j:j\not=i} \lambda_j^{-1}p_j \right)\left(\sum\limits_{j:j\not=i} \lambda_j p_j \right)\nonumber \\
&\geq {1 \over M}\sum\limits_i \lambda_i^{-1}.
\end{align}
Eqn. \eqref{eqn:Second02} implies
\begin{align}\label{eqn:Second03}
\sum\limits_i \lambda_i^{-1}p_i+& {1 \over M}\sum\limits_i \lambda_i^{-1} p_i \left(\sum\limits_{j:j\not=i} \lambda_j^{-1} \right)\left(\sum\limits_k \lambda_k p_k \right)\nonumber \\
&\geq {1 \over M}\sum\limits_i \lambda_i^{-1}.
\end{align}
which rewrites as
\begin{align}\label{eqn:Second04}
\sum\limits_i &\lambda_i^{-1}p_i+ {1 \over M}\left(\sum\limits_i \lambda_i^{-1} p_i \right)\left(\sum\limits_j \lambda_j^{-1} \right)\left(\sum\limits_k \lambda_k p_k \right)\nonumber \\
&\geq {1 \over M}\left(\sum\limits_i \lambda_i^{-2} p_i\right) \left(\sum\limits_k \lambda_k p_k \right)+ {1 \over M}\sum\limits_i \lambda_i^{-1}.
\end{align}

\medskip
\noindent
Finally, we let $p_i={1 \over M}\big({\bf 1}^Tv_i\big)^2$ and substitute the following expressions into \eqref{eqn:Second04}:  
$$\sum\limits_i \lambda_i p_i ={1 \over M}{\bf 1}\!^T \!A_\alpha \!{\bf1}={1 \over M}{\bf 1}\!^T \!\widetilde{C} {\bf1}={c \over M} ~~\text{ as }~a=0,$$
$$\sum\limits_i \lambda_i^{-1} p_i ={1 \over M}{\bf 1}\!^T \!A_\alpha^{-1}\!{\bf1}={1 \over M\,v(\alpha)},$$
$$\sum\limits_i \lambda_i^{-1}=\Tr\big(A_\alpha^{-1}\big),~~\text{and }~\sum\limits_i \lambda_i^{-2} p_i ={1 \over M}{\bf 1}\!^T \!A_\alpha^{-2} {\bf1}.$$
Consequently, \eqref{eqn:Second04} rewrites as
\begin{align}\label{eqn:Second}
M^2&\Big(\!1-v(\alpha)\Tr\big(A_\alpha^{-1}\big)\Big)\! \nonumber \\
&\quad + \!c\Big(\!\Tr\big(A_\alpha^{-1}\big)- v(\alpha) {\bf 1}\!^T \!A_\alpha^{-2}{\bf1}\Big)~\geq 0.
\end{align}

Substituting \eqref{eqn:First} and \eqref{eqn:Second} into \eqref{eqn:qtLaplacian}, we then obtain $\, \triangle  \|C_\alpha\|_F^2 \geq 0$.
\end{proof}

\section*{Acknowledgments}
We would like to thank the anonymous referees for encouraging remarks, valuable feedback, and suggesting ways to improve the paper.
This research was supported by the FAPESP awards 2018/14952-7 and  2018/07826-5, and by the NSF award DMS-1412557.

\bibliographystyle{plain}

\end{document}